\documentclass[10pt, twocolumn]{IEEEtran}

\usepackage{graphicx,epsfig}
\usepackage[noadjust]{cite}
\usepackage{mcite}
\usepackage{amsfonts,helvet}
\usepackage{fancyhdr}
\usepackage{threeparttable}
\usepackage{epsf,epsfig}
\usepackage{amsthm}
\usepackage{amsmath}
\usepackage{siunitx}
\usepackage{amssymb}
\usepackage{dsfont}
\usepackage{subfigure}
\usepackage{color}
\usepackage{algorithm}
\usepackage{algpseudocode}
\usepackage{algcompatible}
\usepackage{enumerate}
\usepackage{cancel}
\usepackage{eucal}

\newtheorem{theorem}{Theorem}

\newtheorem{corollary}{Corollary}

\newtheorem{definition}{Definition}

\newtheorem{lemma}{Lemma}

\newtheorem{remark}{Remark}

\begin{document}

\title{Low Complexity Antenna Selection for Low Target Rate Users in Dense Cloud Radio Access Networks}

\author{Jeonghun~Park, \emph{Student Member, IEEE}
 and~Robert W. Heath Jr., \emph{Felow, IEEE}
\thanks{J. Park and R. W. Heath Jr. are with the Wireless Networking and Communication Group (WNCG), Department of Electrical and Computer Engineering, 
The University of Texas at Austin, TX 78701, USA. (E-mail: $\left\{\right.$jeonghun, rheath$\left\}\right.$@utexas.edu)}
\thanks{
This research was supported by a gift from Huawei Technologies Co. Ltd.}
}


\maketitle \setcounter{page}{1} 


\begin{abstract}
We propose a low complexity antenna selection algorithm for low target rate users in cloud radio access networks.
The algorithm consists of two phases: In the first phase, each remote radio head (RRH) determines whether to be included in a candidate set by using a predefined selection threshold. In the second phase, RRHs are randomly selected within the candidate set made in the first phase. 
To analyze the performance of the proposed algorithm, we model RRHs' and users' locations by a homogeneous Poisson point process, whereby the signal-to-interference ratio (SIR) complementary cumulative distribution function is derived. By approximating the derived expression, an approximate optimum selection threshold that maximizes the SIR coverage probability is obtained. Using the obtained threshold, we characterize the performance of the algorithm in an asymptotic regime where the RRH density goes to infinity. 
The obtained threshold is then modified depending on various algorithm options.
A distinguishable feature of the proposed algorithm is that the algorithm complexity 
keeps constant independent to the RRH density, so that a user is able to connect to a network without heavy computation at baseband units.
\end{abstract}

\section{Introduction}
Cloud radio access networks (C-RANs) \cite{cran_whitep, rost:commac} use distributed RF units called remote radio heads (RRHs), which are connected to a centralized baseband processor unit (BBU) cloud via highspeed fronthaul. 
Due to this structure, a C-RAN has an inherent advantage for improving the network throughput. For instance, on the uplink, it would be optimum to decode jointly the transmitted symbols by aggregating all the received data, which is naturally possible in a C-RAN setting. 
One issue of the centralized decoding is that it requires significant computation complexity. For example, if computation resources in the BBU cloud are statically multiplexed, sharing huge amount of data between BBUs may cause a computational outage \cite{valenti:2014:gc}. 
For a user that wants high-rate communication, this shortcoming is worth to endure since it compensates high data rate. 
Nevertheless, focusing on a user that wants only low-rate communication, this huge computation complexity is unnecessary; thereby such user should avoid this.

As an alternative of the centralized decoding, for users which require low-rate
communication, we consider to segment the processing so that one BBU 
decodes the data for a single (low target rate) user. In this setting, a RRH selection switch is considered as illustrated in Fig.~\ref{switch_model}. The role of the RRH selection switch is to select a RRH and connect it to the BBU dedicated for the corresponding user, e.g., user 1 in Fig.~\ref{switch_model}. User 1 indicates a low target rate user. For this reason, instead of a fixed RRH-BBU pair, the considered C-RAN has a reconfigurable fronthaul structure \cite{Sundaresan:2013}, where a BBU can be flexibly connected to the selected RRHs among the distributed RRHs. By doing this, the decoding complexity per user is manageable.

\begin{figure}[t]
\centerline{\resizebox{1\columnwidth}{!}{\includegraphics{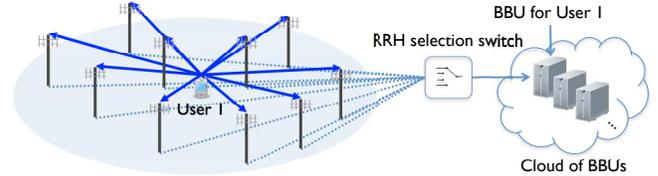}}}   
\caption{The considered uplink C-RAN model. The RRH selection switch selects a RRH and connects it to the BBU dedicated for user 1. User 1 indicates a low target rate user.}  
   \label{switch_model} 
\end{figure}

In the considered C-RAN with a reconfigurable fronthaul, RRH selection is important since it mainly determines the performance. Most prior RRH selection methods demand complexity closely related to the density of the RRH. As a simple example, let's assume that the RRH selection switch selects the nearest RRH to the user. To do this, the RRH selection switch searches all the RRHs and selects the RRH whose the distance is minimum, resulting in that the complexity linearly increases with the RRH density.
When considering a C-RAN with a high RRH density \cite{megen:twc, jun:2014}, this RRH selection method can cause high complexity in the RRH selection switch, where such complexity is not desirable especially for low target rate users.
In this paper, we propose a RRH selection algorithm where
its goal is different from conventional RRH selection algorithms, i.e., keeping reasonable complexity even in dense C-RANs.

\subsection{Related Work}

RRH selection methods in C-RANs was proposed in \cite{liu:2013, shi:2014, fan:2014,  joung:2013, joung:2013gc, joung:2013lett,wan:07}. In \cite{liu:2013}, the downlink sum-rate was characterized as a function of a subset of RRHs 
and based on that a combinatorial optimization problem was formulated to find the optimal subset of RRHs. Similar to \cite{liu:2013}, in \cite{shi:2014}, an optimization problem to select the RRHs was formulated but the optimization goal was minimizing network power consumption. In \cite{fan:2014}, to reduce the complexity caused by estimating instantaneous channel and computing an uplink receiver filter, a channel matrix sparsifying algorithm was proposed for the MMSE receiver. In \cite{joung:2013, joung:2013gc, joung:2013lett}, motivated by the energy efficiency in a large distributed network, energy efficient antenna selection algorithms were proposed.
In \cite{wan:07}, a multi-mode antenna selection algorithm that chooses whether one antenna or multiple antennas for serving one user was proposed.

In another line of research, the signal-to-interference (SIR) coverage probability was characterized when using various cooperation techniques under an assumption of a network modeled by a homogeneous Poisson point process (PPP). For instance, in \cite{mugen:14}, the SIR coverage was analyzed in a uplink C-RAN, where a user is associated with the nearest RRHs.
Assuming a downlink C-RAN where a user is served by multiple RRHs (or base stations (BSs)), the SIR coverage probability was characterized in \cite{NY:dynamic, junzhang:dynamic, zhang:14:spawc, Tanbourgi:2014, 7248946}. Further, by using this characterization, the optimum cluster size was obtained in \cite{NY:dynamic, junzhang:dynamic, zhang:14:spawc}.
In \cite{6909064}, the SIR coverage performance of the rate-splitting with the pair-wise BS cooperation was characterized. 
Considering a multi-tier network, in \cite{6928420, 6879305, 7067349}, a joint transmission method for heterogeneous networks was proposed and the SIR performance was analyzed. 
While the benefits of cooperation was a main topic in \cite{NY:dynamic, junzhang:dynamic, zhang:14:spawc, Tanbourgi:2014, 7248946, 6909064, 6928420, 6879305, 7067349}, \cite{7308016} focused on how each user achieves the benefits of cooperation avoiding the BS conflict problem, which occurs when multiple users want to be served from the same BS.

The main limitation of the existing work \cite{liu:2013, shi:2014, fan:2014, wan:07, joung:2013, joung:2013gc, joung:2013lett, mugen:14, NY:dynamic, 7248946, junzhang:dynamic, zhang:14:spawc, Tanbourgi:2014, 6909064,6928420, 6879305, 7067349} is that a centralized approach is used, where a core processor (e.g., the RRH selection switch in this paper) collects all the information from every RRH such as distances to users or instantaneous channel coefficients for choosing RRHs.
For instance, in \cite{liu:2013}, distances (large scale fading) between RRHs to active users are needed to solve the optimization problem. This approach is not fitted to our aim in a dense C-RAN scenario since the complexity is an increasing function of the RRH density.

\subsection{Contributions}
In this paper, we propose a low complexity RRH selection algorithm. The proposed algorithm consists of two phases. In the first phase, called the distributed selection phase, each RRH compares a distance (or received power) from a user with a predefined selection threshold, and determines whether to be included in a candidate set.
One issue in the first phase is that it is not trivial to extract the required information, e.g., a distance or received power, in each RRH when all the baseband processing such as FFT are placed in the BBU cloud. To resolve this, we assume a LTE channel structure which permits the RRHs to extract the required information from the received signal in the time domain, without performing all the received signal processing found subsequently in the BBU cloud.
In the second phase, called the random selection phase, the RRH selection switch randomly selects RRHs within the candidate set made in the first phase. 
By using two separate phases, the complexity of the algorithm is constant irrespective of the RRH density. 

To analyze the performance of the proposed algorithm, we model a network by using a homogeneous PPP, that allows an expression for the SIR CCDF to be derived in a closed form. Further, for the analytical tractability, we simplify the proposed algorithm so that only one RRH is selected and a distance between each RRH and a user is used in the first phase. Under this assumption, we derive the SIR complementary cumulative distribution function (CCDF) as a function of relevant system parameters, chiefly the selection threshold, the densities of the RRH and the interfering user, the pathloss exponent, and the SIR target.
By approximating the derived SIR CCDF, we find an approximate optimum selection threshold that maximizes the SIR coverage probability.
With the obtained selection threshold, we characterize the SIR coverage probability of the proposed algorithm in an asymptotic regime, and reveal a condition that the relative performance loss caused by the random selection vanishes.
Then, we modify the obtained approximate optimum selection threshold so as to work for a general case of the algorithm, i.e., if multiple RRHs are selected or received power is used. 


The remainder of the paper is organized as follows. In Section II, 
the proposed algorithm and the system model used in the paper are explained. In Section III, an approximate optimum selection threshold is obtained and the performance of the proposed algorithm is characterized.
In Section IV, the obtained selection threshold is modified for various algorithm options and Section V concludes the paper.

\section{System Model }
In this section, we first explain the proposed algorithm and the RRH setting for applying the algorithm in practice. 
Next, we introduce the network model and the RRH selection model for analyzing the performance of the proposed algorithm.

\subsection{RRH Selection Algorithm}

In this subsection, we explain the proposed RRH selection algorithm. 
For applying the algorithm, we consider a general uplink cellular system implemented by a C-RAN, where single-antenna RRHs are distributed and connect to centralized BBUs. 
The location of the $i$-th RRH is denoted as ${\bf{d}}_i$. The set of RRH locations are denoted as $\Phi = \{{\bf{d}}_i, i \in \mathbb{N}\}$.
Single-antenna users transmit the uplink data through the network. We denote 
that the $i$-th user is located at ${\bf{u}}_i$, and the set of the users' locations is $\Phi_{\rm u} = \{{\bf{u}}_i, i \in \mathbb{N}\}$.
We only focus on user 1 located at ${\bf u}_1$ since the algorithm can be applied for each user equivalently. Without loss of generality, we assume ${\bf{ u}}_1 = {\bf{0}}$. This assumption can be generalized easily by shifting the location of each RRH by ${\bf{d}}_i - {\bf{u}}_1 \; \forall i$.
It is worthwhile to mention that the applicability of the proposed algorithm is not restricted by a particular network model, such as a PPP network model.

The proposed algorithm consists of two phases, called the distributed selection phase and the random selection phase, respectively.

\subsubsection{Phase 1-Distributed Selection}
The goal of this phase is to determine a candidate set of RRHs. 
To do this, each RRH compares the distance from user 1 with a predefined selection threshold $R_{\rm th}$. 
Denoting a candidate set as $\CMcal{A}$, a RRH whose a distance from user 1 is less than $R_{\rm th}$ will be in $\CMcal{A}$. In other words, ${\bf{d}}_i \in \CMcal{A}$ if $\left\| {\bf{d}}_i \right\| < R_{\rm th}$. 
Clearly, every RRH included in $\CMcal{A}$ has a distance less then $R_{\rm th}$, i.e., $\mathop {\max} \limits_{{\bf{d}}_i \in \CMcal{A}} \left\| {\bf{d}}_i \right\| < R_{\rm th}$. 
If the selection threshold $R_{\rm th}$ is too small, then all the RRHs' distances are larger than a threshold, i.e., $\left\| {\bf{d}}_i \right\| > R_{\rm th}$ for $\forall {\bf{d}}_i \in \Phi$, the candidate set is empty. In this case, user 1 fails to connect to a RRH and the outage occurs. 
Instead of a distance, each RRH also can use received power. Given a received power threshold $P_{\rm th}$, the RRH whose received power larger than $P_{\rm th}$ will be in $\CMcal{A}$, and otherwise the RRH will not be included in $\CMcal{A}$. We denote that $\left| \CMcal{A} \right| = M$, $M\ge 0$. 
\subsubsection{Phase 2-Random Selection}
In this phase, the RRH selection switch randomly selects RRHs within the candidate set made in the distributed selection phase. A set of selected RRHs in this phase is denoted as $\CMcal{B}$, where $\left| \CMcal{B} \right| = L$, $L \ge 1$ and $M \ge L$. 
For instance, assuming that $\CMcal{A} =\{{\bf{d}}_{1}, \cdots {\bf{d}}_{M} \}$, we have $\CMcal{B} = \{{\bf{d}}_{i_1}, \cdots {\bf{d}}_{i_L} \}$ with probability $1/\binom{M}{L}$ for any $\{i_1,...,i_L\} \subseteq \{1,...,M\}$ since the RRHs are chosen randomly. When only one RRH is selected in the RRH selection switch, i.e., $L=1$, one RRH is selected within $\CMcal{A}$ with probability $1/M$.
Fig.~\ref{algo_des} illustrates the proposed algorithm assuming $M=2$ and $L=1$

One point about the random selection phase is that there is a non-zero possibility that more than two users select the same RRH. When assuming that the RRHs are densely deployed, which is the case we focus on, this probability becomes small. 

In the distributed selection phase, since each of RRH performs the comparison in a distributed way, no centralized processing is required. For this reason, 
the complexity is independent to the density of the RRH. In the random selection phase, the RRH selection switch randomly selects a RRH within the candidate set, so that the complexity is also independent to the density of the RRH.
To show this specifically, we consider a simple random delay method that can be used in the random selection phase. If the RRH is included in $\CMcal{A}$ in the distributed selection phase, the RRH sends a predefined $1$-bit symbol through the optical fiber to the RRH selection switch, otherwise sends nothing. Before sending the signal, the RRH generates a random delay and sends the symbol after the generated delay time. Then, the RRH selection switch selects the RRH whose the sent symbol arrives at the first time. By doing this, the RRH selection switch is able to complete the RRH selection within a constant time. 

For more clarification, we compare the random selection phase and the nearest RRH selection. We assume that the same distributed selection phase is used for refining the candidate set $\CMcal{A}$, but in the nearest RRH selection, the RRH selection switch chooses the RRH whose the distance is the minimum in $\CMcal{A}$.
To do this, the RRHs included in $\CMcal{A}$ sends the measured distance after quantizing it through the optical fiber. Here, the required number of quantization bit is obviously more than $1$. Then, the RRH selection switch collects all the distance information sent from each RRH and finds the minimum, which needs $\left| \CMcal{A} \right|$ complexity. If the cardinality of $\CMcal{A}$ increases, the complexity for choosing the RRH should also increase.
On the contrary to this, in the random selection, the selection switch does not have to collect all the distance information from each RRH in $\CMcal{A}$, which makes the selection complexity constant irrespective of $\left| \CMcal{A} \right|$. 
Finally, the proposed algorithm is able to select a RRH  for each user with a constant complexity.

\begin{figure}[t]
\centerline{\resizebox{0.8\columnwidth}{!}{\includegraphics{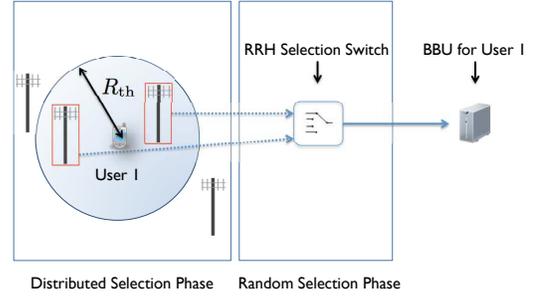}}}   
\caption{Illuratration of the proposed RRH selection algorithm when distance information is used. We assume that $M=2$, $L=1$. In the distributed selection phase, the candidate set is determined. In the figure, the RRHs marked by the red rectangle are included in $\CMcal{A}$. In the random selection phase the RRH selection switch randomly chooses one RRH in $\CMcal{A}$ with probability $1/2$.}  
   \label{algo_des} 
\end{figure}

\subsection{RRH Setting}
Since the proposed algorithm requires a distance or received power, each RRH should extract this information from the signals it receives. When all the digital processing such as the FFT as used in 3GPP LTE are placed in the BBU cloud, however, it is not clear that how each RRH can extract the information required for the proposed algorithm. For example, without the FFT, each RRH should obtain the required information from only the time domain signals.
In this subsection, we explain how each RRH obtains the required information by using the characteristic of the LTE channel structure.
Before data transmission, a user sends the random access preamble signal generated from the Zadoff-Chu sequence through the physical random access channel (PRACH) for initial access. There can be two kinds of interference to this preamble signal. The first one is from users that are actively communicating with their selected RRHs. These signals are transmitted through the physical uplink shared channel (PUSCH) \cite{ghosh:lte}. Conventionally, this can be eliminated easily by using the FFT due to the orthogonal property of the OFDM, though this cannot be applied due to the lack of the FFT in each RRH. The second one is from users that are transmitting other preamble signals through the physical random access channel (PRACH). 
 
We first remove the interference on the PUSCH. Since the PUSCH and the PRACH are defined to be separately placed in the frequency domain \cite{ghosh:lte}, each RRH uses bandpass filter (BPF) that only allows to pass the frequency band corresponding to the PRACH. Then, the signals on the PUSCH are removed. 
Through this method, the only remaining signals are the preamble signals transmitted on the PRACH. The characteristic of the preamble signals is that they are generated from Zadoff-Chu sequence, and each user has different root of the sequence. Two key properties of Zadoff-Chu sequence are as follows: It is a constant amplitude zero auto-correlation (CAZAC) sequence, and it preserves its property of a CAZAC sequence in both of the time and frequency domain. Due to these properties, the preamble signals on PRACH from the different users 
have zero cross correlation each other in the time domain \cite{ghosh:lte}. For this reason, each RRH discriminates a preamble signal transmitted from a particular user by applying the conventional technique as in the LTE standard, i.e., multiplying the preamble signals with the Zadoff-chu sequence assigned to the particular user.
As a result, the only remaining signal is the preamble signal transmitted from the particular user due to the zero cross correlation property. 
Then the RRH extracts the required information from the remaining signal.
Fig.~\ref{bpf_zadoff} describes the whole procedures of how to obtain the required information without using the FFT at each RRH. 

\begin{figure}[t] 
\centerline{\resizebox{0.89\columnwidth}{!}{\includegraphics{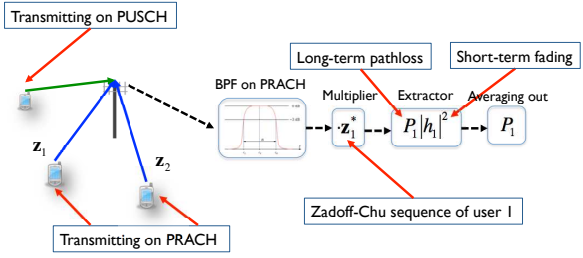}}}    
\caption{The description of how to extract the required information without using frequency domain processing. Eliminating the signals on the PUSCH by using the bandpass filter, and discriminate the preamble signals by multiplying Zadoff-Chu sequence assigned to each user. By repeating this process and averaging the received power, each RRH obtains the averaged received power.}
\label{bpf_zadoff}
\end{figure}

Now we explain the model mainly used for analyzing the performance of the proposed algorithm.

\subsection{Network Model}
We consider a network modeled by a homogeneous PPP, so that a RRH's location ${\bf{d}}_i\;\forall i$ is distributed according to a homogeneous PPP with density $\lambda$. Each user's location ${\bf{u}}_i\;\forall i$ is also distributed as a homogeneous PPP with density $\lambda_{\rm u}$. 
We do not assume power control as this would depend on the RRH that is eventually selected.

\subsection{RRH Selection Model}
Henceforth, for analytical tractability, we assume that the proposed algorithm uses a distance in the distributed selection phase and only one RRH is selected ($L=1$). This assumption is generalized later. Once the RRH is selected, it is connected to the BBU dedicated for the user 1, and the BBU decodes the uplink data symbol by using the received signal at the selected RRH. Extensions are possible to support MIMO, e.g., to use multiple co-located RRHs from the tower. 

\subsection{Signal Model}
The BBU performs single user detection by treating other users' interference as noise \cite{geng:tin}. This is a reasonable assumption since we assume that one BBU is allocated for handling one user. 
Thanks to the property of a homogeneous PPP, we are able to assume that ${\bf{u}}_1 = {\bf{0}}$ without loss of generality. Denoting the index of the selected RRH for user 1 as ${\rm s}$, the received signal at the selected RRH is given by
\begin{align} \label{def_sig}
y_{\rm s} = \left\| {\bf{d}}_{\rm s} \right\|^{-\beta/2} h_{{\rm s},1} x_1 + \sum_{{\bf{u}}_i \in \Phi_{\rm u} \backslash {\bf{u}}_1} {
\left\| {\bf{d}}_{\rm s} - {\bf{u}}_i \right\|^{-\beta/2}h_{{\rm s},i} x_i} + z_{\rm s},
\end{align}
where ${\bf{d}}_{\rm s}$ is a location of the selected RRH , $h_{i,j} \sim \CMcal{CN}\left(0, 1\right)$ is a Rayleigh fading coefficient from user $j$ to RRH $i$, and $z_{\rm s} \sim {\CMcal{CN}} \left(0, \sigma^2 \right)$ is additive white Gaussian noise and 
$\beta$ is the pathloss exponent. The uplink symbol transmitted from user $i$ is indicated by $x_i$, whose $\mathbb{E}\left[ \left| x_i\right|^2\right] = 1$.  

Now we define the instantaneous SIR CCDF. 
Denoting $H_{i,j} = \left| h_{i,j} \right|^2$, the instantaneous SIR CCDF is defined as 
\begin{align} \label{def_sir}
{\mathtt{P}}\left(\theta, \lambda, \lambda_{\rm u}, \beta \right) = \mathbb{P}\left[ \frac{\left\| {\bf{d}}_{\rm s} \right\|^{-\beta} H_{\rm s, 1}}{\sum_{{\bf{u}}_i \in \Phi_{\rm u}\backslash {\bf{u}}_1 }   \left\| {\bf{d}}_{\rm s} - {\bf{u}}_i \right\|^{-\beta} H_{{\rm s}, i} } > \theta \right],
\end{align}
where $\theta$ is the SIR target. The noise term is neglected for analytical tractability. The noise term can be incorporated with more complicated calculations, but it makes it hard to devise intuition from the expression.

\section{Performance Characterization}
In this section, we provide analytical results on the performance of the proposed algorithm. At first, we characterize the SIR coverage probability. Then we optimize a predefined selection threshold by using the obtained SIR coverage expression. Finally, we characterize the performance of the proposed algorithm in an asymptotic regime.

\subsection{SIR CCDF Characterization}
In this subsection, we derive the SIR CCDF if a RRH is selected by using the proposed selection algorithm. First, we obtain the probability density function (PDF) of the $\left\| {\bf{d}}_{\rm s} \right\|$ in the following Lemma.
\begin{lemma} \label{lem_pdf}
Given the selection threshold $R_{\rm th}$, the PDF of the random variable $\left\| {\bf{d}}_{\rm s} \right\|$ is 
\begin{align}
f_{\left\| {\bf{d}}_{\rm s}\right\|}^{R_{\rm th}}\left(r \right) = \frac{2r}{R_{\rm th}^2} ,\;{\rm for}\;0< r< R_{\rm th}.
\end{align}
\end{lemma}
\begin{proof}
Denote the number of RRHs inside the closed set $S \subseteq \mathbb{R}^2$ as $N\left(S \right)$.
When $N\left(\CMcal{B}\left(0, R_{\rm th} \right) \right) = K$, the conditional PDF of $\left\| {\bf{d}}_{\rm s}\right\|$ is
\begin{align} \label{lem_pdf_1}
f_{\left\| {\bf{d}}_{\rm s}\right\|}^{R_{\rm th}}\left(\left. r \right| N\left(\CMcal{B}\left(0, R_{\rm th} \right) \right) = K \right) = \frac{2r}{R_{\rm th}^2} ,\;{\rm for}\;0< r< R_{\rm th}.
\end{align}
This is because in a homogeneous PPP conditioned on the number of points in $\CMcal{B}\left(0, R_{\rm th} \right)$, points in $\CMcal{B}\left(0, R_{\rm th} \right)$ are independently and uniformly distributed in the bounded set $\CMcal{B}\left(0, R_{\rm th} \right)$. Marginalizing \eqref{lem_pdf_1} for $K$, 
\begin{align} \label{pdf_d_s_not_normalized}
&f_{\left\| {\bf{d}}_{\rm s}\right\|}^{R_{\rm th},{\rm Not-normalized}}\left(r \right) \nonumber \\
& = \mathbb{E}\left[f_{\left\| {\bf{d}}_{\rm s}\right\|}^{R_{\rm th}}\left(\left. r \right| N\left(\CMcal{B}\left(0, R_{\rm th} \right) \right) = K \right) \right]\nonumber \\
&= \frac{2r}{R_{\rm th}^2} \sum_{K = 1}^{\infty} \frac{\left(\lambda \pi R_{\rm th}^2 \right)^K}{K !} e^{-\lambda \pi R_{\rm th}^2} \nonumber \\
& \mathop = \limits^{(a)} \frac{2r}{R_{\rm th}^2}(1-e^{-\lambda \pi R_{\rm th}^2}),\;{\rm for}\;0< r< R_{\rm th},
\end{align}
where (a) follows that 
\begin{align}
&\sum_{K=1}^{\infty} \mathbb{P}\left[N\left(\CMcal{B}\left(0, R_{\rm th} \right) \right)=K \right] = \nonumber \\
&\sum_{K = 0}^{\infty} \frac{\left(\lambda \pi R_{\rm th}^2 \right)^K}{K !} e^{-\lambda \pi R_{\rm th}^2}  - e^{-\lambda \pi R_{\rm th}^2} = 1 - e^{-\lambda \pi R_{\rm th}^2},
\end{align}
by the second axiom of probability. Normalizing \eqref{pdf_d_s_not_normalized} so that the total probability is equal to $1$, we have
\begin{align}
f_{\left\| {\bf{d}}_{\rm s}\right\|}^{R_{\rm th}}\left(r \right) & = \frac{2r}{R_{\rm th}^2},\;{\rm for}\;0< r< R_{\rm th}.
\end{align}
\end{proof} 
Leveraging Lemma \ref{lem_pdf}, the SIR CCDF \eqref{def_sir} is derived in Theorem \ref{theo_sirccdf}. 
\begin{theorem} \label{theo_sirccdf}
Given the selection threshold $R_{\rm th}$, the instantaneous SIR CCDF is 
\begin{align} \label{th_lowccdf_claim}
&{\mathtt{P}}\left(R_{\rm th}, \theta, \lambda, \lambda_{\rm u}, \beta \right)  \nonumber \\
& = \left(1 - e^{-\lambda \pi R_{\rm th}^2} \right) \frac{\left( 1 - e^{-\pi \lambda_{\rm u} \theta^{2/\beta} \frac{1}{{\rm sinc} \left(2/\beta \right)}R_{\rm th}^2  }\right)}{\pi \lambda_{\rm u} \theta^{2/\beta} \frac{1}{{\rm sinc} \left(2/\beta \right)} R_{\rm th}^2}.
\end{align}
\end{theorem}
\begin{proof}
Since the outage occurs when a candidate set $\CMcal{A}$ is empty, we only consider the case that $\CMcal{A}$ is not empty. The SIR \eqref{def_sir} is rewritten as
\begin{align} \label{th_sir_1}
& {\mathtt{P}}\left(R_{\rm th},\theta,  \lambda, \lambda_{\rm u}, \beta \right) 
\nonumber \\
& = \mathbb{P}\left[ \CMcal{A} \ne \emptyset \right]
\mathbb{P} \left[\left. \frac{ \left\| {\bf{d}}_{\rm s} \right\|^{-\beta}H_{\rm s,1}}{\sum_{{\bf{u}}_i \in \Phi_{\rm u} \backslash {\bf{u}}_1 }   \left\| {\bf{d}}_{\rm s} - {\bf{u}}_i \right\| ^{-\beta} H_{{\rm s}, i}} > \theta \right| \CMcal{A} \neq \emptyset \right]
\nonumber \\
& \mathop = \limits^{} \mathbb{P}\left[ \CMcal{A} \ne \emptyset \right]
\mathbb{P} \left[ H_{{\rm s},1} > \left\| {\bf{d}}_{\rm s} \right\|^{\beta}\theta  \sum_{{\bf{u}}_i \in \Phi_{\rm u} \backslash {\bf{u}}_1}   \left\| {\bf{d}}_{\rm s} - {\bf{u}}_i \right\| ^{-\beta} H_{{\rm s}, i} \right]
\nonumber \\
& \mathop = \limits^{(a)} \mathbb{P}\left[ \CMcal{A} \ne \emptyset \right] \mathbb{E}\left[e^{ - \left\| {\bf{d}}_{\rm s} \right\|^{\beta} \theta \sum_{{\bf{u}}_i \in \Phi_{\rm u}  \backslash {\bf{u}}_1 }   \left\| {\bf{d}}_{\rm s} - {\bf{u}}_i \right\| ^{-\beta} H_{{\rm s}, i} } \right]
\nonumber \\
&\mathop = \limits^{(b)} \left(1 - e^{-\lambda \pi R_{\rm th}^2} \right) \mathbb{E}\left[e^{ - \left\| {\bf{d}}_{\rm s} \right\|^{\beta} \theta \sum_{{\bf{u}}_i \in \Phi_{\rm u} \backslash {\bf {u}}_1}   \left\| {\bf{d}}_{\rm s} - {\bf{u}}_i \right\| ^{-\beta} H_{{\rm s}, i} } \right],
\end{align}
where (a) comes from that $H_{{\rm s}, i}$ for $i \in \mathbb{N}$ follows the exponential distribution with unit mean and (b) follows 
\begin{align}
\mathbb{P}\left[ \CMcal{A} \ne \emptyset \right] &= 1 - \mathbb{P}\left[N\left(\CMcal{B}\left(0, R_{\rm th} \right) \right) = 0\right] \nonumber \\
&= \left(1 - e^{-\lambda \pi R_{\rm th}^2} \right).
\end{align}
We now calculate the expectation in \eqref{th_sir_1}. 
First,
\begin{align} \label{th1_sir_2}
&\mathbb{E}\left[e^{-\left\| {\bf{d}}_{\rm s} \right\|^{\beta} \theta \sum_{{\bf{u}}_i \in \Phi_{\rm u} \backslash {\bf{u}}_1}   \left\| {\bf{d}}_{\rm s} - {\bf{u}}_i \right\| ^{-\beta} H_{{\rm s}, i} } \right] \nonumber \\
&= \mathbb{E}_{{\bf{d}}_{\rm s}}\left[\mathbb{E}\left[\left.e^{-\left\| {\bf{d}}_{\rm s} \right\|^{\beta} \theta \sum_{{\bf{u}}_i \in \Phi_{\rm u} \backslash {\bf{u}}_1}   \left\| {\bf{d}}_{\rm s} - {\bf{u}}_i \right\| ^{-\beta} H_{{\rm s}, i} } \right| {\bf{d}}_{\rm s} \right] \right] \nonumber \\
& \mathop = \limits^{(a)} \mathbb{E}_{{\bf{d}}_{\rm s}}\left[\mathbb{E}\left[\left.e^{-\left\| {\bf{d}}_{\rm s} \right\|^{\beta} \theta \sum_{{\bf{u}}_i \in \Phi_{\rm u} \backslash {\bf{u}}_1 }   \left\| {\bf{u}}_i \right\| ^{-\beta} H_{{\rm s}, i} } \right| {\bf{d}}_{\rm s} \right] \right] \nonumber \\
& \mathop = \mathbb{E}_{{\bf{d}}_{\rm s}}\left[ \CMcal{L}_{I} \left(\left\| {\bf{d}}_{\rm s} \right\|^{\beta}\theta \right)\right].
\end{align}
where (a) follows the stationarity of a homogeneous PPP and Slivnyak's theorem \cite{stoyan:book}.
$\CMcal{L}_{ I}  \left(s\right)$ is the Laplace functional of $I$, where $I = \sum_{{\bf{u}}_i \in \Phi_{\rm u} \backslash {\bf{u}}_1 }   \left\|  {\bf{u}}_i \right\| ^{-\beta} H_{{\rm s}, i} $. $\CMcal{L}_{ I}  \left(s\right)$ is derived as 
\begin{align} \label{laplace_I}
\CMcal{L}_I\left(s \right) =  \exp\left(-\pi \lambda_{\rm u} s^{2/\beta}\frac{1}{{\rm sinc}\left(2/\beta \right)} \right),
\end{align}
where the detailed proof is in \cite{baccelli:06}. Plugging \eqref{laplace_I} into \eqref{th1_sir_2}, 
\begin{align}
&\mathbb{E}_{{\bf{d}}_{\rm s}}\left[ \CMcal{L}_{ I  } \left(\left\| {\bf{d}}_{\rm s} \right\|^{\beta}\theta \right)\right] \nonumber \\
&= \mathbb{E}\left[\exp\left(-\pi \lambda_{\rm u} \theta^{2/\beta} \left\| {\bf{d}}_{\rm s} \right\| ^2\frac{1}{{\rm sinc}\left(2/\beta \right)} \right) \right] \nonumber \\
& \mathop = \limits^{(a)} \int_{0}^{R_{\rm th}} \exp\left(-\pi \lambda_{\rm u} \theta^{2/\beta} \left\| {\bf{d}}_{\rm s} \right\| ^2\frac{1}{{\rm sinc}\left(2/\beta \right)} \right) \frac{2r}{R_{\rm th}}{\rm d} r \nonumber \\
& \mathop = \limits^{} \frac{1 - \exp\left(-\pi \lambda_{\rm u} \theta^{2/\beta} \frac{1}{{\rm sinc} \left(2/\beta \right)}R_{\rm th}^2  \right)}{\pi \lambda_{\rm u} \theta^{2/\beta} \frac{1}{{\rm sinc} \left(2/\beta \right)} R_{\rm th}^2},
\end{align}
where (a) comes from Lemma \ref{lem_pdf}. This completes the proof.
\end{proof}
\begin{figure}[t] 
\centerline{\resizebox{0.8\columnwidth}{!}{\includegraphics{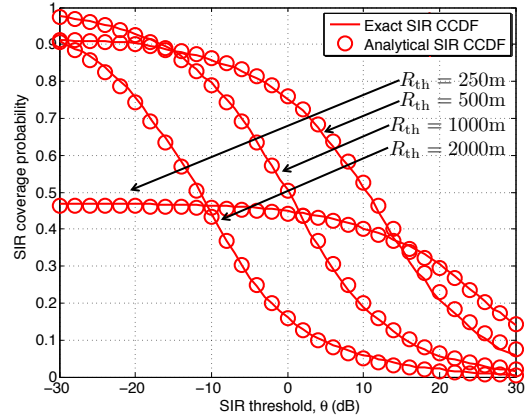}}}   
\caption{The SIR CCDF verification when $\lambda = 10^{-5}/\pi$, $\lambda_{\rm u} = 10^{-6}/\pi$, and $\beta = 4$.
}
\label{fig_verify}
\end{figure}

The obtained SIR CCDF is verified in Fig.~\ref{fig_verify}. As observed, the derived SIR CCDF tightly matches with the exact SIR CCDF obtained by Monte-Carlo simulations over entire range of $\theta$.
Fig.~\ref{fig_verify} also gives intuition of how the selection threshold $R_{\rm th}$ affects the SIR coverage performance. Applying the proposed selection algorithm, there are two cases of outage. The first case is when no RRH is in the candidate set, i.e., $\CMcal{A} = \emptyset$. The second case is when the SIR is lower than $\theta$.
Now we see examples for each case of outage depending on $R_{\rm th}$.
When $R_{\rm th} = 250{\rm m}$, the SIR CCDF has a plateau when $\theta < 0{\rm dB}$. This is mainly because the possibility of the event $\CMcal{A} = \emptyset$ is too high, therefore the SIR CCDF is dominated by the first case of outage. In contrast, with $R_{\rm th} = 2000{\rm m}$, the selected RRH is likely to be far from the user since the selection threshold is too large, resulting in that the SIR coverage performance degrades severely when $\theta$ increases, i.e., the SIR CCDF is dominated by the second case of outage. This observation implies that the selection threshold should be optimized depending on the system parameters, e.g., $\lambda$, $\lambda_{\rm u}$, $\theta$, and $\beta$.

\subsection{Selection Threshold Optimization}

\begin{figure}[t] 
\centerline{\resizebox{0.8\columnwidth}{!}{\includegraphics{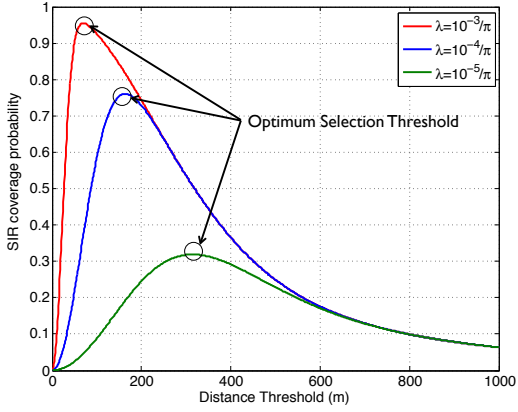}}}    
\caption{The SIR coverage probability with parameter sweaping for $R_{\rm th}$. It is assumed that $\lambda_{\rm u} = 10^{-5}/\pi$, $\beta = 4$, and $\theta = 0{\rm dB}$.
}
\label{para_sweep}
\end{figure}

In this subsection, we derive an approximate optimal selection threshold $\tilde{R}_{\rm th}^{\star}$ to maximize the SIR coverage performance given system parameters $\theta, \lambda$, and $\lambda_{\rm u}$. 
For intuition, we first illustrate the SIR coverage performance depending on $R_{\rm th}$ in Fig.~\ref{para_sweep}. As observed in Fig.~\ref{para_sweep}, the optimum selection threshold $R_{\rm th}^{\star}$ exists, and also the SIR coverage performance has a sharp shape around the $R_{\rm th}^{\star}$ especially when the RRHs are densely deployed, so that there can be significant performance loss when using wrong $R_{\rm th}$. 

Obtaining the exact $R_{\rm th}^{\star} = \arg \max \mathtt{P}\left(R_{\rm th} \right)$, however, is challenging. Specifically, there is no closed form solution satisfying 
\begin{align}
\frac{\partial \mathtt{P}\left(R_{\rm th} \right)}{\partial R_{\rm th}} = 0.
\end{align}
For this reason, we rather use an approximate SIR CCDF to obtain the optimum selection threshold.
Lemma \ref{lem_approx} gives an approximation of $\mathtt{P}\left(R_{\rm th} \right)$.
\begin{lemma} \label{lem_approx}
The SIR CCDF \eqref{th_lowccdf_claim} is approximated by 
\begin{align}
&\tilde {\mathtt{P}}\left(R_{\rm th},\theta,  \lambda, \lambda_{\rm u}, \beta \right)  \nonumber \\
&= \frac{\lambda \pi R_{\rm th}^2}{\left(1 + \lambda \pi R_{\rm th}^2\right)\left(1+\pi \lambda_{\rm u} \theta^{2/\beta}\frac{1}{{\rm sinc}\left(2/\beta \right)} R_{\rm th}^2\right)}.
\end{align}
\end{lemma}
\begin{proof}
From the Taylor expansion of the exponential function $e^x = 1 + x/1! + x^2/2! + \cdots$, we have
\begin{align}
e^{-x} = \frac{1}{1 + x/1 + x^2/2! + \cdots} \approx \frac{1}{1+x}.
\end{align}
Using this approximation, the first part of the SIR CCDF \eqref{th_lowccdf_claim} is approximated as
\begin{align} \label{lem_approx_1}
1 - e^{-\lambda \pi R_{\rm th}^2} \approx \frac{\lambda \pi R_{\rm th}^2}{1 + \lambda \pi R_{\rm th}^2},
\end{align}
and the second part of \eqref{th_lowccdf_claim} is also approximated as
\begin{align} \label{lem_approx_2}
1 - e^{-\pi \lambda_{\rm u} \theta^{2/\beta}\frac{1}{{\rm sinc}\left(2/\beta \right)} R_{\rm th}^2} \approx \frac{\pi \lambda_{\rm u} \theta^{2/\beta}\frac{1}{{\rm sinc}\left(2/\beta \right)} R_{\rm th}^2}{1 + \pi \lambda_{\rm u} \theta^{2/\beta}\frac{1}{{\rm sinc}\left(2/\beta \right)} R_{\rm th}^2}
\end{align}
Plugging \eqref{lem_approx_1} and \eqref{lem_approx_2} into \eqref{th_lowccdf_claim}, we complete the proof.
\end{proof}
By leveraging Lemma \ref{lem_approx}, Corollary \ref{coro_opt} provides an approximate optimum selection threshold.
\begin{corollary} \label{coro_opt}
Given $\theta, \lambda$, and $\lambda_{\rm u}$, an approximate optimal selection threshold that maximizes $\tilde {\mathtt{P}}\left(R_{\rm th},\theta,  \lambda, \lambda_{\rm u}, \beta \right) $ is 
\begin{align} \label{opt_R}
\tilde{R}_{\rm th}^{\star} = \left(\frac{1}{\pi^2 \lambda  \lambda_{\rm u}  \theta^{2/\beta}\frac{1}{\rm sinc\left(2/\beta \right)}}\right)^{\frac{1}{4}}.
\end{align}
\end{corollary}
\begin{proof}
To find $\tilde{R}_{\rm th}^{\star} = \arg \max \tilde {\mathtt{P}}\left(R_{\rm th},\theta,  \lambda, \lambda_{\rm u}, \beta \right) $, we solve
\begin{align}
\frac{\partial \tilde{\mathtt{P}}\left(R_{\rm th} \right) }{\partial R_{\rm th}} = 0,
\end{align} 
where
\begin{align}
&\tilde {\mathtt{P}}\left(R_{\rm th},\theta,  \lambda, \lambda_{\rm u}, \beta \right)  \nonumber \\
&= \frac{\lambda \pi R_{\rm th}^2}{\left(1 + \lambda \pi R_{\rm th}^2\right)\left(1+\pi \lambda_{\rm u} \theta^{2/\beta}\frac{1}{{\rm sinc}\left(2/\beta \right)} R_{\rm th}^2\right)}.
\end{align}
It has a closed-form solution
\begin{align} \label{coro_opt_complex}
R_{\rm th} = \left(\frac{1}{\pi^2 \lambda  \lambda_{\rm u}  \theta^{2/\beta}\frac{1}{\rm sinc\left(2/\beta \right)}}\right)^{\frac{1}{4}},
\end{align}
which completes the proof.
\end{proof}

\begin{figure}[t] 
\centerline{\resizebox{0.78\columnwidth}{!}{\includegraphics{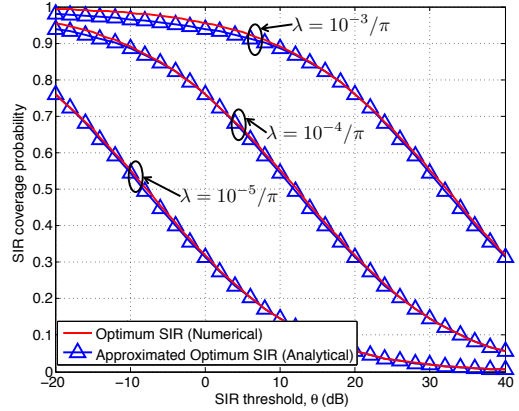}}}    
\caption{The comparison between the optimal SIR CCDF obtained numerically and the approximated optimum SIR CCDF obtained analytically. It is assumed that $\lambda_{\rm u} = 10^{-5}/\pi$, $\beta = 4$, and $\lambda \in \{10^{-3}/\pi, 10^{-4}/\pi, 10^{-5}/\pi\}$. }  \label{fig_compare}  
\end{figure}

\begin{remark}
\normalfont
As observed in Corollary \ref{coro_opt}, $\tilde R^{\star}_{\rm th}$ is inversely proportional to $\lambda$, $\lambda_{\rm u}$, and $\theta^{2/\beta}$. This gives intuition for deciding a selection threshold $R_{\rm th}$. When the network is dense, (large $\lambda$), $R_{\rm th}$ should be small since the probability that there is a RRH located close to a user is high. When there are many interfering users (large $\lambda_{\rm u}$), $R_{\rm th}$ should also be small since there is only little chance of successfully communicating with the selected RRH if the RRH is located far from the user due to the large interference. This is also true when $\theta$ increases. This intuition agrees with the observations from Fig.~\ref{fig_verify}. 
\end{remark}

To demonstrate the obtained selection threshold, we compare the SIR CCDF with $\tilde R^{\star}_{\rm th}$ and the numerically obtained optimum SIR CCDF. For the numerically obtained optimum SIR CCDF, we calculate all the SIR coverage probability for $R_{\rm th} \in \left[0, \infty\right)$ and pick the maximum one. Fig.~\ref{fig_compare} shows the comparison between them. The SIR CCDF with $\tilde R^{\star}_{\rm th}$ is reasonably close to the numerically obtained optimum SIR CCDF over entire range of $\theta$.

\subsection{Asymptotic Performance Analysis}

In this subsection, we analyze the performance of the proposed algorithm in the asymptotic regime, where $\lambda \rightarrow \infty$ and $\lambda_{\rm u} \rightarrow \infty$.
Comparing to the SIR performance of the nearest RRH selection, which corresponds to the best case in the same assumption with the proposed algorithm, a relative performance loss is defined.
Then, we reveal the necessary condition that makes the performance loss vanish in the asymptotic regime.
First, we derive the SIR CCDF when user 1 selects the nearest RRH to the origin. 
\begin{lemma} \label{theo_nearest_claim}
When user 1 selects the nearest RRH, the instantaneous SIR CCDF is 
\begin{align} \label{th_nearest}
{\mathtt{P}}_{\rm n}\left(\theta, \lambda, \lambda_{\rm u}, \beta \right)  
 = \frac{\lambda \ {\rm sinc}\left(\frac{2}{\beta} \right)}{\lambda_{\rm u}\theta^{2/\beta} + \lambda \ {\rm sinc}\left( \frac{2}{\beta}\right)}.
\end{align}
\end{lemma}
\begin{proof}
Rewriting the definition of the SIR CCDF \eqref{def_sir},
\begin{align} \label{th_nearest_sir1}
& {\mathtt{P}}_{\rm n}\left(\theta, \lambda, \lambda_{\rm u}, \beta \right) 
\nonumber \\
& = \mathbb{P} \left[ \frac{ \left\| {\bf{d}}_{\rm s} \right\|^{-\beta}H_{{\rm s}, 1}}{\sum_{{\bf{u}}_i \in \Phi_{\rm u} \backslash {\bf{u}}_1 }   \left\| {\bf{d}}_{\rm s} - {\bf{u}}_i \right\| ^{-\beta} H_{{\rm s}, i}} > \theta \right]
\nonumber \\
&= \mathbb{P} \left[ H_{{\rm s}, 1} > \left\| {\bf{d}}_{\rm s} \right\|^{-\beta}\theta  \sum_{{\bf{u}}_i \in \Phi_{\rm u} \backslash {\bf{u}}_1 }   \left\| {\bf{d}}_{\rm s} - {\bf{u}}_i \right\| ^{-\beta} H_{{\rm s}, i} \right] \nonumber \\
&= \mathbb{E}\left[e^{ - \left\| {\bf{d}}_{\rm s} \right\|^{\beta} \theta \sum_{{\bf{u}}_i \in \Phi_{\rm u} \backslash {\bf{u}}_1}   \left\| {\bf{d}}_{\rm s} - {\bf{u}}_i \right\| ^{-\beta} H_{{\rm s}, i} } \right],
\end{align}
where $\left\| {\bf{d}}_{\rm s} \right\| \le \left\| {\bf{d}}_{i} \right\|$ for $i \in \mathbb{N}$ since we assume that the nearest RRH selection. 
Similar to Theorem \ref{theo_sirccdf},
\begin{align} \label{theo_near_laplace}
&\mathbb{E}\left[e^{-\left\| {\bf{d}}_{\rm s} \right\|^{\beta} \theta \sum_{{\bf{u}}_i \in \Phi_{\rm u} \backslash {\bf{u}}_1 }   \left\| {\bf{d}}_{\rm s} - {\bf{u}}_i \right\| ^{-\beta} H_{{\rm s}, i} } \right] \nonumber \\
& \mathop = \mathbb{E}_{{\bf{d}}_{\rm s}}\left[ \CMcal{L}_{I  }  \left(\left\| {\bf{d}}_{\rm s} \right\|^{\beta}\theta \right) \right],
\end{align}
where 
$\CMcal{L}_{I } \left(s \right)$ is the Laplace functional of $I$ where $I = \sum_{{\bf{u}}_i \in \Phi_{\rm u} \backslash {\bf{u}}_1}   \left\|  {\bf{u}}_i \right\| ^{-\beta} H_{{\rm s}, i} $. 
The Laplace functional of $I$ is given by
\begin{align}
\CMcal{L}_I\left(s \right) = \exp\left(-\pi \lambda_{\rm u} s^{2/\beta}\frac{1}{{\rm sinc}\left(2/\beta \right)} \right).
\end{align}
Plugging it into \eqref{theo_near_laplace}, we have
\begin{align} \label{th_near_expectation}
&\mathbb{E}_{{\bf{d}}_{\rm s}}\left[ \CMcal{L}_{ I  } \left(\left\| {\bf{d}}_{\rm s} \right\|^{\beta}\theta \right) \right] \nonumber \\
&= \mathbb{E}\left[\exp\left(-\pi \lambda_{\rm u} \theta^{2/\beta} \left\| {\bf{d}}_{\rm s} \right\| ^2\frac{1}{{\rm sinc}\left(2/\beta \right)} \right) \right].
\end{align}
Now we use the PDF of $\left\| {\bf{d}}_{\rm s} \right\|$, which is \cite{andrew:10}
\begin{align} \label{th_near_pdf}
f_{\left\| {\bf{d}}_{\rm s} \right\| }\left(r \right) = {2\lambda \pi r} e^{-\lambda \pi r^2}.
\end{align}
Leveraging \eqref{th_near_pdf}, the expectation in \eqref{th_near_expectation} is calculated as
\begin{align}
&\mathbb{E}\left[\exp\left(-\pi \lambda_{\rm u} \theta^{2/\beta} \left\| {\bf{d}}_{\rm s} \right\| ^2\frac{1}{{\rm sinc}\left(2/\beta \right)} \right) \right] \nonumber \\
&= \int_{0}^{\infty} \!\! \exp\left(-\pi \lambda_{\rm u} \theta^{2/\beta} \frac{\left\| {\bf{d}}_{\rm s} \right\| ^2}{{\rm sinc}\left(2/\beta \right)} \right) 2 \lambda \pi r \exp\left(-\lambda \pi r^2 \right)  {\rm d} r
\nonumber \\
&=\frac{\lambda \ {\rm sinc}\left(\frac{2}{\beta} \right)}{\lambda_{\rm u}\theta^{2/\beta} + \lambda \ {\rm sinc}\left( \frac{2}{\beta}\right)},
\end{align}
which completes the proof.
\end{proof}
Now we define the relative SIR performance loss in the following.

\begin{definition}
The relative SIR performance of the proposed algorithm compared to the nearest RRH selection is defined as
\begin{align}
&L(R_{\rm th},\theta,  \lambda, \lambda_{\rm u}, \beta) = \frac{{\mathtt{P}}\left(R_{\rm th},\theta,  \lambda, \lambda_{\rm u}, \beta \right) }{{\mathtt{P}}_{\rm n}\left(\theta, \lambda, \lambda_{\rm u}, \beta \right) } \nonumber \\
& = \frac{\left(1 - e^{-\lambda \pi R_{\rm th}^2} \right) \frac{\left( 1 - e^{-\pi \lambda_{\rm u} \theta^{2/\beta} \frac{1}{{\rm sinc} \left(2/\beta \right)}R_{\rm th}^2  }\right)}{\pi \lambda_{\rm u} \theta^{2/\beta} \frac{1}{{\rm sinc} \left(2/\beta \right)} R_{\rm th}^2}}{\frac{\lambda \ {\rm sinc}\left(\frac{2}{\beta} \right)}{\lambda_{\rm u}\theta^{2/\beta} + \lambda \ {\rm sinc}\left( \frac{2}{\beta}\right)}} \nonumber \\
& = \frac{\left(1 - e^{-\lambda \pi R_{\rm th}^2} \right)\left( 1 - e^{-\pi \lambda_{\rm u} \theta^{2/\beta} \frac{1}{{\rm sinc} \left(2/\beta \right)}R_{\rm th}^2  }\right) }{\frac{\pi \lambda_{\rm u} \lambda \theta^{2/\beta} R_{\rm th}^2 }{\lambda_{\rm u}\theta^{2/\beta} + \lambda \ {\rm sinc}\left( \frac{2}{\beta}\right) }  }.
\end{align}
Since the nearest RRH selection is the best case of the proposed algorithm, $L(R_{\rm th},\theta,  \lambda, \lambda_{\rm u}, \beta) \le 1$. When $L(R_{\rm th},\theta,  \lambda, \lambda_{\rm u}, \beta)=1$, the proposed algorithm has the same SIR performance with the nearest RRH selection. When the proposed algorithm uses the approximate optimum selection threshold $\tilde R_{\rm th}^{\star}$, we denote that
\begin{align}
&L(\tilde{R}_{\rm th}^{\star},\theta,  \lambda, \lambda_{\rm u}, \beta) \nonumber \\
&= \frac{\left(1 - e^{-\lambda \pi \left( \tilde R_{\rm th}^{\star} \right)^2} \right)\left( 1 - e^{-\pi \lambda_{\rm u} \theta^{2/\beta} \frac{1}{{\rm sinc} \left(2/\beta \right)}\left( \tilde R_{\rm th}^{\star} \right)^2  }\right)  }{\frac{\pi \lambda_{\rm u} \lambda \theta^{2/\beta} \left( \tilde R_{\rm th}^{\star} \right)^2}{\lambda_{\rm u}\theta^{2/\beta} + \lambda \ {\rm sinc}\left( \frac{2}{\beta}\right) }   }.
\end{align}
\end{definition}
The performance loss ($L(R_{\rm th},\theta,  \lambda, \lambda_{\rm u}, \beta) \le 1$) in the proposed algorithm comes from the random selection phase, which cannot guarantee the nearest RRH is selected. As mentioned before, however, the nearest RRH selection requires complexity that increases with the RRH density. For this reason, the performance loss is interpreted as a cost for keeping the complexity independent to the RRH density.

Now, we reveal the condition for $L(\tilde{R}_{\rm th}^{\star},\theta,  \lambda, \lambda_{\rm u}, \beta) \rightarrow 1$ in the asymptotic regime, i.e., $\lambda \rightarrow \infty$ and $\lambda_{\rm u} \rightarrow \infty$ in the following theorem.
\begin{theorem} \label{asym_anal}
Assuming that $\lambda \rightarrow \infty$ and $\lambda_{\rm u} \rightarrow \infty$, the performance loss vanishes, i.e., $L(\tilde{R}_{\rm th}^{\star},\theta,  \lambda, \lambda_{\rm u}, \beta) \rightarrow 1$ if
\begin{align}
\frac{\sqrt{\lambda}}{\sqrt{\lambda_{\rm u}}} \rightarrow \infty.
\end{align}
\end{theorem}
\begin{proof}
From \eqref{opt_R}, $L(\tilde{R}_{\rm th}^{\star},\theta,  \lambda, \lambda_{\rm u}, \beta)$ is
\begin{align}
&L(\tilde{R}_{\rm th}^{\star},\theta,  \lambda, \lambda_{\rm u}, \beta) \nonumber \\
&= \frac{\left(1 - e^{\left(-{\frac{\sqrt{\lambda}\sqrt{{\rm sinc}\left(2/\beta \right)}}{\sqrt{\lambda_{\rm u}}  \theta^{1/\beta} }} \right)}\right)\left(1 - e^{\left(-\frac{  \sqrt{\lambda_{\rm u} } \theta^{1/\beta } }{\sqrt{\lambda}\sqrt{{\rm sinc}\left(2/\beta \right)}} \right)} \right)}{\frac{\sqrt{\lambda}  \sqrt{\lambda_{\rm u}} \theta^{1/\beta} \sqrt{{\rm sinc}\left(2/\beta \right)} }{\left(\lambda_{\rm u}\theta^{2/\beta} + \lambda \ {\rm sinc}\left(2/\beta \right) \right)}}.
\end{align}
Letting $C_{\lambda/\lambda_{\rm u}} = \sqrt{\lambda}/\sqrt{\lambda_{\rm u}}$,  
\begin{align}
&L(\tilde{R}_{\rm th}^{\star},\theta,  \lambda, \lambda_{\rm u}, \beta) \nonumber \\
&= \frac{\left(1 - e^{\left(-C_{\lambda/\lambda_{\rm u}}{\frac{\sqrt{{\rm sinc}\left(2/\beta \right)}}{ \theta^{1/\beta} }} \right)}\right)\left(1 - e^{\left(-\frac{1}{C_{\lambda/\lambda_{\rm u}}}\frac{  \theta^{1/\beta } }{\sqrt{{\rm sinc}\left(2/\beta \right)}} \right)} \right)}{\frac{C_{\lambda/ \lambda_{\rm u}}}{\left( \frac{\theta^{1/\beta}}{\sqrt{{\rm sinc}\left(2/\beta \right)}} \right)}  }  \nonumber \\
&+\frac{\left(1 - e^{\left(-C_{\lambda/\lambda_{\rm u}}{\frac{\sqrt{{\rm sinc}\left(2/\beta \right)}}{ \theta^{1/\beta} }} \right)}\right)\left(1 - e^{\left(-\frac{1}{C_{\lambda/\lambda_{\rm u}}}\frac{  \theta^{1/\beta } }{\sqrt{{\rm sinc}\left(2/\beta \right)}} \right)} \right)}{\frac{1  }{C_{\lambda/ \lambda_{\rm u}}\left( \frac{\sqrt{{\rm sinc}\left(2/\beta \right)}}{\theta^{1/\beta}} \right)}}.
\end{align}
Now consider a function defined as
\begin{align}
f\left(x\right) =  \frac{\left(1 - e^{-x} \right)\left(1 - e^{-\frac{1}{x}} \right)}{x} +  \frac{\left(1 - e^{-x} \right)\left(1 - e^{-\frac{1}{x}} \right)}{1/x},
\end{align}
where $x = C_{\lambda/\lambda_{\rm u}}{\frac{\sqrt{{\rm sinc}\left(2/\beta \right)}}{ \theta^{1/\beta} }}$.
When $x \rightarrow \infty$, we have
\begin{align}
&\lim_{x \rightarrow \infty} f\left( x \right)  \nonumber \\
&= \lim_{x \rightarrow \infty} \frac{\left(1 - e^{-x} \right)\left(1 - e^{-\frac{1}{x}} \right)}{x} +  \frac{\left(1 - e^{-x} \right)\left(1 - e^{-\frac{1}{x}} \right)}{1/x} \nonumber \\
&= \lim_{x \rightarrow \infty} \frac{\left(1 - e^{-x} \right)\left(1 - e^{-\frac{1}{x}} \right)}{1/x} \nonumber \\
&= \lim_{x \rightarrow \infty} x - xe^{-\frac{1}{x}} - xe^{-x}+xe^{-x-\frac{1}{x}} \nonumber \\
&= \lim_{y \rightarrow 0} \frac{\left(1 - e^{-y} \right)}{y} \nonumber \\
&\mathop = \limits^{(a)}  1,
\end{align}
where (a) follows L'Hopital's rule.
This concludes that when $C_{\lambda/\lambda_{\rm u}} \rightarrow \infty$, $L(\tilde{R}_{\rm th}^{\star},\theta,  \lambda, \lambda_{\rm u}, \beta)\rightarrow 1$, i.e., the performance loss vanishes. This completes the proof.
\end{proof}

\begin{figure}[t!] 
\centerline{\resizebox{0.75\columnwidth}{!}{\includegraphics{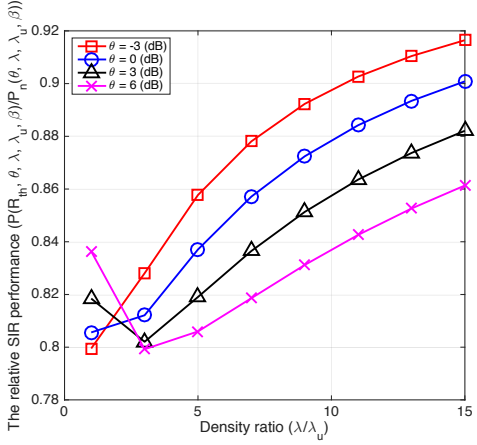}}}    
\caption{The SIR performance loss of the proposed algorithm depending on the density ratio and various SIR thresholds. It is assumed that $\beta = 4$. As observed in the figure, the performance loss due to the random selection decreases as the density ratio increases.}
\label{compare_nonasym}
\end{figure}

In Theorem \ref{asym_anal}, it is shown that the performance loss due to the random selection vanishes if the relative RRH density $C_{\lambda/\lambda_{\rm u}}$ goes to infinity. In practice, however, the assumption $C_{\lambda/\lambda_{\rm u}} \rightarrow \infty$ is too extreme even in a dense C-RAN scenario. Nevertheless, Theorem \ref{asym_anal} is still valid to get intuition of the relative SIR performance of the proposed algorithm in the non-asymptotic regime, i.e, $\frac{\sqrt{\lambda}}{\sqrt{\lambda_{\rm u}}} < \infty$. For instance, since $L(\tilde{R}_{\rm th}^{\star},\theta,  \lambda, \lambda_{\rm u}, \beta)$ is smooth for all $C_{\lambda/\lambda_{\rm u}}$, it is obvious that $L(\tilde{R}_{\rm th}^{\star},\theta,  \lambda, \lambda_{\rm u}, \beta)$ approaches to $1$ as $C_{\lambda/\lambda_{\rm u}}$ increases. In other words, the performance loss due to the random selection in the proposed algorithm becomes negligible as the density of RRH increases. 
One noticeable point here is that the algorithm complexity keeps constant independent to the RRH density.
This intuition is demonstrated by the simulation in Fig.~\ref{compare_nonasym}.
In Fig.~\ref{compare_nonasym}, $L(\tilde{R}_{\rm th}^{\star},\theta,  \lambda, \lambda_{\rm u}, \beta)$ is drawn depending on the density ratio $\frac{\lambda}{\lambda_{\rm u}}$ and the SIR threshold $\theta$. As observed in the figure, $L(\tilde{R}_{\rm th}^{\star},\theta,  \lambda, \lambda_{\rm u}, \beta)$ increases as $\frac{\lambda}{\lambda_{\rm u}}$ increases, and also $L(\tilde{R}_{\rm th}^{\star},\theta,  \lambda, \lambda_{\rm u}, \beta)$ is larger at low $\theta$.
From this observation, we see that the proposed algorithm performs well (i) in a dense RRH environment and (ii) at a low SIR threshold. 
One non-trivial observation in Fig.~\ref{compare_nonasym} is that at high SIR thresholds $\theta = 3, 6 \;({\rm dB})$, there is a range of $\frac{\lambda}{\lambda_{\rm u}}$ where $L(\tilde{R}_{\rm th}^{\star},\theta,  \lambda, \lambda_{\rm u}, \beta)$ and $\frac{\lambda}{\lambda_{\rm u}}$ inversely proportional, i.e., $L(\tilde{R}_{\rm th}^{\star},\theta,  \lambda, \lambda_{\rm u}, \beta)$ decreases as $\frac{\lambda}{\lambda_{\rm u}}$ increases. This is particularly observed in the low density ratio region $\frac{\lambda}{\lambda_{\rm u}} < 3$. 
The implication behind this observation is that in the region $1<\frac{\lambda}{\lambda_{\rm u}} < 3$, the SIR performance of the nearest RRH selection improves faster than that of the proposed algorothm, resulting in $L(\tilde{R}_{\rm th}^{\star},\theta,  \lambda, \lambda_{\rm u}, \beta)$ rather decreases.

\section{Generalization}
In this section, we generalize the approximate optimum selection threshold derived in the previous section. First, we relax the assumption where a distance between each RRH and a user is used as a predefined threshold. Next, the assumption that only one RRH is selected is generalized to multiple RRHs, i.e., $L>1$. 

\subsection{Received Power Threshold}

In a real wireless environment, it is difficult to estimate the exact distance from the user because of long-term fading such as a shadowing. For this reason, it is more desirable for the proposed RRH selection algorithm to use the received power to select a RRH rather than the distance. 
To derive an approximate optimum selection threshold analytically, we make an assumption about a shadowing.
After averaging out the fast fading coefficients, the received power at RRH $i$ from user 1 is 
\begin{align} \label{rx_power}
P_{i,1} = S_{i,1} \left\| {\bf{d}}_i \right\|^{-\beta},
\end{align}
where $S_{i,1}$ denotes a shadowing coefficient from user 1 and RRH $i$. 
Each RRH measures the received power, and compares it with a predefined received power threshold $P_{\rm th}$. If the measured received power is larger than the threshold, i.e.,
$P_{i,1} > P_{\rm th}$, the RRH located at ${\bf{d}}_i$ is included in a candidate set $\CMcal{A}$, unless it is not. The second phase of the algorithm is equivalent with a case where the distance threshold is used. 

Now we attempt to obtain the optimum selection threshold for a case where the received power threshold is used.
To this end, Lemma \ref{lem_lattice} is introduced. It is for incorporating a shadowing effect into the approximate optimum distance threshold derived in Corollary \ref{coro_opt}.

\begin{lemma} \label{lem_lattice}
Assume generic shadowing coefficient, denoted as $S$ where
\begin{align} \label{shadowing_cond}
\mathbb{E}\left[S^{\frac{2}{\beta}} \right] < \infty.
\end{align}
Then, the process of propagation losses experienced by the typical user is an non-homogeneous Poisson process on $\mathbb{R}^+$ with intensity measure
\begin{align} \label{lem_measure}
\Lambda\left(\left[0, t\right) \right) &= \lambda \int_{\mathbb{R}^2} \mathbb{P} \left[ \frac{x^{\beta}}{S} \in \left[0,t \right) \right]{\rm d}x\nonumber \\
&= \lambda \pi \mathbb{E}\left[S^{\frac{2}{\beta}} \right] t^{\frac{2}{\beta}}. 
\end{align}
\end{lemma}
\begin{proof}
See the reference \cite{Blas:2013} Lemma 1.
\end{proof}


By leveraging Lemma \ref{lem_lattice}, we obtain an approximate optimum received power threshold $\tilde P^{\star}_{\rm th}$ in the following corollary.  

\begin{corollary}
Assume independent and identical generic shadowing coefficient denoted as $S$ satisfying \eqref{shadowing_cond}. 
When the received power is used as the selection threshold, an approximate optimum selection threshold $\tilde P^{\star}_{\rm th}$ is
\begin{align}
\tilde P^{\star}_{\rm th} = \frac{\left({\tilde R^{\star}_{\rm th, sh}}\right)^{-\beta}}{\mathbb{E}\left[S^{\frac{2}{\beta}} \right]^{-\frac{\beta}{2}} },
\end{align}
where $\tilde R^{\star}_{\rm th, sh}$ is defined as
\begin{align}
\tilde{R}_{\rm th, sh}^{\star} = \left(\frac{1}{\pi^2 \lambda  \lambda_{\rm u}  \theta^{2/\beta}\frac{E_S}{\rm sinc\left(2/\beta \right)}}\right)^{\frac{1}{4}},\; {\rm with}
\end{align}
\begin{align}
E_S = {\mathbb{E}\left[{S^{2/\beta}} \right]} \mathbb{E}\left[\frac{1}{S}\right]^{2/\beta}.
\end{align}
\end{corollary}
\begin{proof}
Rewriting the SIR CCDF with the arbitrary shadowing assumption, 
\begin{align} \label{sir_ccdf_shadowing_pre}
&{\mathtt{P}}\left(R_{\rm th}, \theta, \lambda, \lambda_{\rm u}, \beta \right)  \nonumber \\
&= \mathbb{P}\left[ \CMcal{A} \ne \emptyset \right] \! \mathbb{P}\left[\left. \frac{\left\| {\bf{d}}_{\rm s} \right\|^{-\beta} H_{{\rm s},1}S_{{\rm s}, 1}}{\sum_{{\bf{u}}_i \in \Phi_{\rm u}\backslash {\bf{u}}_1 }  \! \left\| {\bf{d}}_{\rm s} - {\bf{u}}_i \right\|^{-\beta} \! H_{{\rm s}, i} S_{{\rm s}, i}} \! > \! \theta \right| \! \CMcal{A} \neq \emptyset \! \right],
\end{align}
where $S_{j,i}$ for $j,i \in \mathbb{N}$ is a shadowing coefficient between RRH $j$ and user $i$, while the index ${\rm s}$ means the index of the selected RRH. For instance, $S_{{\rm s}, i}$ means a shadowing coefficient between the selected RRH and the user $i$. Subsequently we have
\begin{align} \label{sir_ccdf_shadowing_late}
\eqref{sir_ccdf_shadowing_pre}  &=  \mathbb{P}\left[ \CMcal{A} \ne \emptyset \right] \cdot \nonumber \\
&\;\;\;\; \mathbb{E}\left[ \mathbb{E}\left[\left. e^{ - \left\| {\bf{d}}_{\rm s} \right\|^{\beta} \theta \sum_{{\bf{u}}_i \in \Phi_{\rm u}  \backslash {\bf{u}}_1 }   \left\| {\bf{d}}_{\rm s} - {\bf{u}}_i \right\| ^{-\beta} H_{{\rm s}, i} S_{{\rm s}, i}/S_{{\rm s},1}}\right| S_{\rm s,1} \right] \right]
\nonumber \\
&\mathop \ge \limits^{(a)} \left(1 - e^{-\lambda \pi R_{\rm th}^2} \right) \cdot \nonumber \\
&\;\;\;\;    \mathbb{E}\left[e^{ - \left\| {\bf{d}}_{\rm s} \right\|^{\beta} \theta \sum_{{\bf{u}}_i \in \Phi_{\rm u} \backslash {\bf {u}}_1}   \left\| {\bf{d}}_{\rm s} - {\bf{u}}_i \right\| ^{-\beta} H_{{\rm s}, i} S_{{\rm s}, i} \mathbb{E}\left[1/{S_{{\rm s}, 1}} \right] } \right]  \nonumber \\
&\mathop = \limits^{(b)} \left(1 - e^{-\lambda \pi R_{\rm th}^2} \right) \frac{\left( 1 - e^{-\pi \lambda_{\rm u} \theta^{2/\beta} \frac{E_S}{{\rm sinc} \left(2/\beta \right)}R_{\rm th}^2  }\right)}{\pi \lambda_{\rm u} \theta^{2/\beta} \frac{E_S}{{\rm sinc} \left(2/\beta \right)} R_{\rm th}^2},
\end{align}
where (a) comes from Jensen's inequality and (b) follows that the Laplace functional of $I' =  \sum_{{\bf{u}}_i \in \Phi_{\rm u} \backslash {\bf{u}}_1 }   \left\|  {\bf{u}}_i \right\| ^{-\beta} H_{{\rm s}, i} S_{{\rm s}, i} \mathbb{E}\left[ 1/{S_{{\rm s}, 1}} \right]$ given  as 
\begin{align} \label{laplace_int_shadowing}
\CMcal{L}_{I'}\left(s \right) =  \exp\left(-\pi \lambda_{\rm u} s^{2/\beta}
{\mathbb{E}\left[{S^{2/\beta}} \right]} \mathbb{E}\left[\frac{1}{S}\right]^{2/\beta} \frac{1}{{\rm sinc}\left(2/\beta \right)} \right),
\end{align}
with $S$ is a generic shadowing coefficient.
Defining $E_S = {\mathbb{E}\left[{S^{2/\beta}} \right]} \mathbb{E}\left[\frac{1}{S}\right]^{2/\beta}$, we obtain \eqref{sir_ccdf_shadowing_late}.
Now, we first obtain an approximate optimum distance threshold $R_{\rm th}$ that maximizes a lower bound on the SIR CCDF with a shadowing assumption \eqref{sir_ccdf_shadowing_late}. This can be calculated directly by using Corollary \ref{coro_opt}. 
\begin{align} \label{dist_thres_opt_shadowing}
\tilde{R}_{\rm th, sh}^{\star} = \left(\frac{1}{\pi^2 \lambda  \lambda_{\rm u}  \theta^{2/\beta}\frac{E_S}{\rm sinc\left(2/\beta \right)}}\right)^{\frac{1}{4}}.
\end{align}
Since this is a distance threshold, we now obtain a received power threshold that provides the equivalent performance with the obtained $\tilde{R}_{\rm th, sh}^{\star}$ \eqref{dist_thres_opt_shadowing}. 
${\mathtt{P}}\left(R_{\rm th}, \theta, \lambda, \lambda_{\rm u}, \beta \right)$ can be represented as
\begin{align} \label{p_represent_exp}
{\mathtt{P}}\left(R_{\rm th}, \theta, \lambda, \lambda_{\rm u}, \beta \right) =  \left(1 - e^{- M'} \right) \frac{\left( 1 - e^{- \frac{\lambda_{\rm u}}{\lambda} \theta^{2/\beta} \frac{E_S}{{\rm sinc} \left(2/\beta \right)}M'  }\right)}{\frac{ \lambda_{\rm u}}{\lambda} \theta^{2/\beta} \frac{E_S}{{\rm sinc} \left(2/\beta \right)} M'},
\end{align}
where $M' = \lambda \pi R_{\rm th}^2$ is the average number of the selected RRHs in the distributed selection phase when the distance threshold is $R_{\rm th}$. 
From \eqref{p_represent_exp}, it is reasonable to interpret that the SIR coverage performance is determined by the average number of the selected RRHs in the distributed selection phase, i.e., $M'$. With $\tilde{R}_{\rm th, sh}^{\star}$, the average number of the selected RRHs in the distributed selection phase is characterized as
\begin{align}
\lambda \pi \left(\tilde  R^{\star}_{\rm th, sh}\right) ^2 &= \lambda \int_{\mathbb{R}^2} \mathbb{P}\left[x < \tilde  R^{\star}_{\rm th, sh} \right] {\rm d}x \nonumber \\
&= \lambda \int_{\mathbb{R}^2} \mathbb{P}\left[ \left( \tilde  R^{\star}_{\rm th, sh}\right)^{-\beta} < x^{-\beta} \right] {\rm d}x.
\end{align}
Now, rewriting \eqref{lem_measure}, 
\begin{align}
\Lambda\left(\left[0, t\right) \right) &= \lambda \int_{\mathbb{R}^2} \mathbb{P} \left[ \frac{x^{\beta}}{S} \in \left[0,t \right) \right]{\rm d}x \nonumber \\
&\mathop = \limits^{(a)} \lambda \int_{\mathbb{R}^2} \mathbb{P} \left[ \frac{1}{t} < {S}{x^{-\beta}} \right]{\rm d}x \label{coro2_p_thres} \\
&\mathop = \limits^{(b)} \lambda \pi \mathbb{E}\left[S^{\frac{2}{\beta}} \right] t^{\frac{2}{\beta}} \label{coro2_t},
\end{align}
where (a) follows the non-negativity of the received power and (b) follows Lemma \ref{lem_lattice}. From \eqref{coro2_p_thres}, we find that the intensity measure $\Lambda\left(\left[0,t \right) \right)$ is the average number of the selected RRHs in the distributed selection phase when the received power threshold is 
\begin{align}
P_{\rm th} = \frac{1}{t}.
\end{align}
To have the same SIR coverage performance with the case where the proposed algorithm uses an approximate optimum distance threshold $\tilde R^{\star}_{\rm th, sh}$, the received power threshold $1/t$ should satisfy
\begin{align}
\lambda \pi \left(\tilde  R^{\star}_{\rm th, sh}\right) ^2 = \lambda \pi \mathbb{E}\left[S^{\frac{2}{\beta}} \right] t^{\frac{2}{\beta}},
\end{align}
which provides
\begin{align}
t = \frac{\left({\tilde R^{\star}_{\rm th, sh}}\right)^{\beta}}{\mathbb{E}\left[S^{\frac{2}{\beta}} \right]^{\frac{\beta}{2}}}.
\end{align}
This completes the proof.
\end{proof}


\begin{remark}
\normalfont
When $S=1$ (no shadowing assumption), $\tilde R^{\star}_{\rm th, sh} = \tilde R^{\star}_{\rm th}$ and $\tilde P^{\star}_{\rm th}$ boils down to
\begin{align}
\tilde P^{\star}_{\rm th} = \left({\tilde R^{\star}_{\rm th}}\right)^{-\beta},
\end{align}
which is equivalent with the pathloss of the distance threshold ${\tilde R^{\star}_{\rm th}}$.
For this reason, using ${\tilde R^{\star}_{\rm th}}$ and $\tilde P^{\star}_{\rm th}$ provide the equivalent SIR coverage performance in this case. 
\end{remark}


\subsection{Multiple RRHs Selection}
In this subsection, we assume that a user selects multiple RRHs to improve the SIR performance. Under this assumption, we find an appropriate selection threshold that improves the SIR coverage probability for a case of $L>1$. For simplicity, we assume that $L \le M$.
When $L$ RRHs are selected in the RRH selection switch, the BBU uses maximum ratio combining (MRC) technique to boost the desired signal power. Denoting the indices of the selected RRHs as ${\rm s}_1, ..., {\rm s}_{L}$, the SIR CCDF is
\begin{align} \label{sir_multi_rrh}
{\mathtt{P}}_{m}\left(\theta, \lambda, \lambda_{\rm u}, \beta \right) = \mathbb{P}\left[ \frac{ \left| \sum_{\ell =1}^{L} \left\| {\bf{d}}_{\rm s_{\ell}} \right\|^{-\beta/2} {H_{\rm s_{\ell},1}} \right|^2}{ \sum_{\ell=1}^{L}  I_{\ell}} > \theta \right],
\end{align}
where $I_{{\ell}} = {\sum_{{\bf{u}}_i \in \Phi_{\rm u}\backslash {\bf{u}}_1 }   \left\| {\bf{d}}_{\rm s_{\ell}} - {\bf{u}}_i \right\|^{-\beta} H_{{\rm s}_{\ell}, i} H_{{\rm s}_{\ell}, 1} } $. Note that ${\bf{d}}_{{\rm s}_{\ell}}$ for $\ell = 1,...,L$ is the location of the selected RRH. Unlike the previous case where only one RRH is selected, however, it is not straightforward to compute the SIR CCDF. 
For this reason, instead of exact characterization, we provide an approximation of  \eqref{sir_multi_rrh}.

Using the proposed algorithm, the distribution of $\left\|{\bf{d}}_{{\rm s}_{\ell}} \right\|$ for $\ell = 1,...,L$ is identical since there is no dependency in locations when selecting RRHs. For this reason, in an average sense, MRC would provide a $L$-fold array gain to the desired signal. The aggregated interference power $I_{\ell}$ for $\ell = 1,...,L$ is also identically distributed due to the stationarity of a homogeneous PPP.
By approximating identically distributed random variables with the same random variables, we get the following expression.
\begin{align} \label{sir_ccdf_multiple_approx}
{\mathtt{P}}_{m}\left(\theta, \lambda, \lambda_{\rm u}, \beta \right) &\approx \tilde {\mathtt{P}}_{m}\left(\theta, \lambda, \lambda_{\rm u}, \beta \right) \nonumber \\
&= \mathbb{P}\left[ \frac{L  \left\| {\bf{d}}_{\rm s_{1}} \right\|^{-\beta} H_{\rm s_{1},1} }{ I_{1}} > \theta \right]
\end{align}
The approximated SIR CCDF \eqref{sir_ccdf_multiple_approx} is equivalent with that of the single RRH selection case, when $L$-fold array gain is provided to the desired signal power. This is an equivalent benefit as reducing the target SIR $\theta$ to $\theta/L$.
Modifying the obtained selection threshold $\tilde R_{\rm th}^{\star}$ \eqref{opt_R} with $\theta/L$, we have
\begin{align} \label{opt_r_multi}
\tilde R_{\rm th, multi}^{\star} = \left(\frac{L^{2/\beta}}{\pi^2 \lambda  \lambda_{\rm u}  \theta^{2/\beta}\frac{1}{\rm sinc\left(2/\beta \right)}}\right)^{\frac{1}{4}}.
\end{align}
In \eqref{opt_r_multi}, we observe that when the number of selected RRHs increases, the selection threshold also increases by $L^{\frac{1}{2\beta}}$. The explanation of this observation is as follows: when $ R_{\rm th}$ is too large, the outage occurs mainly because the selected RRH is located too far from a user, so that the pathloss is too large. When a user can select multiple RRHs, however, a user has other chances to select different RRHs, probably located more closer to the user. For this reason, the selection threshold becomes larger when the number of selected RRHs increases.

\section{Conclusions}
In this paper, we proposed a low complexity RRH selection algorithm for a low target rate user in a dense C-RAN. 
By using the two separate phases, each of which performs the distributed selection and the random selection, the algorithm complexity is kept constant, and does not depend on the RRH density. 
For the performance analysis, we modeled a network by a homogenous PPP. By using tools of stochastic geometry, we derived the SIR CCDF of the proposed algorithm. From the obtained SIR CCDF expression, we obtained the approximate optimum selection threshold $\tilde R^{\star}_{\rm th}$ that maximizes the SIR CCDF of the proposed algorithm. The simulation results demonstrates that the obtained selection threshold provided the performance close to the optimum SIR coverage performance obtained numerically. 
We also revealed a condition that the relative performance loss coming from the random selection vanishes in an asymptotic regime.
Generalizing the algorithm, the obtained $\tilde R^{\star}_{\rm th}$ was modified to $\tilde P^{\star}_{\rm th}$ or $\tilde R^{\star}_{\rm th, multi}$ for a case where received power is used as a selection threshold or multiple RRHs are selected in the algorithm.

The key feature of the proposed algorithm is that it has complexity independent to the RRH density, so that the RRH selection switch can keep its complexity reasonable irrespective of the RRH density. Due to this, the performance loss is inevitable using the proposed algorithm. 
As the RRH density increases, however, this performance loss becomes negligible.
Future work could be directed to incorporate more advanced cooperation algorithm with the proposed RRH selection. 

\bibliographystyle{IEEEtran}
\bibliography{ref}

\end{document}